\documentclass[10pt,conference]{IEEEtran}
\usepackage{amsmath}%
\usepackage{amsfonts}%
\usepackage{amssymb}%
\usepackage{graphicx}%

\newtheorem{theorem}{Theorem}
\newenvironment{proof}[1][Proof]{\textbf{#1.} }{\ \rule{0.5em}{0.5em}}

\begin{document}

\title{Power Control for Maximum Throughput in Spectrum Underlay Cognitive Radio Networks}


\author{
\authorblockN{John Tadrous, Ahmed Sultan, Mohammed Nafie and Amr El-Keyi}
\authorblockA{Wireless Intelligent Networks Center (WINC)\\
Nile University, Cairo, Egypt\\ 
Email: john.tadrous@nileu.edu.eg, \{asultan, mnafie, aelkeyi\}@nileuniversity.edu.eg}}
\maketitle

\begin{abstract}
We investigate power allocation for users in a spectrum underlay cognitive network. Our objective is to find a power control scheme that allocates transmit power for both primary and secondary users so that the overall network throughput is maximized while maintaining the quality of service (QoS) of the primary users greater than a certain minimum limit. Since an optimum solution to our problem is computationally intractable, as the optimization problem is non-convex, we propose an iterative algorithm based on sequential geometric programming, that is proved to converge to at least a local optimum solution. We use the proposed algorithm to show how a spectrum underlay network would achieve higher throughput with secondary users operation than with primary users operating alone. Also, we show via simulations that the loss in primary throughput due to the admission of the secondary users is accompanied by a reduction in the total primary transmit power. \footnote{This work was supported in part by a grant from the Egyptian National Telecommunications Regulatory Authority
} 
\end{abstract}

\section{Introduction}
\label{sec:Intro}

Cognitive radios, or secondary users, are spectrum agile devices that can sense the spectrum, identify the underutilized bands, and opportunistically share those bands with the original users (primary users) such that they do not cause harmful interference on those primary users \cite{FCC}. Hence, more users can share the spectrum efficiently and provide more utilization to the underutilized bands. However, when a secondary user is severely interfering with the primary users, it should abandon this primary band and try to exploit another underutilized one.

Cognitive radios can operate with a primary network using either overlay, or underlay spectrum sharing \cite{Ian}. In spectrum overlay, secondary users are allowed to access the primary bands that are totally free of primary operation. Consequently, the interference caused on primary users is almost negligible, and robust sensing is demanded of secondary users. In spectrum underlay, secondary users can share the same spectrum band concurrently with the primary users, so long as the interference caused on the primary users is small enough to guarantee minimum primary QoS. The secondary users do not have to employ sophisticated sensing mechanisms. Instead, efficient resource allocation procedure is needed so that the spectrum can be shared with high utilization.

In this paper we consider the problem of power allocation for primary and secondary users in a spectrum underlay system. Most of the existing work separates the power allocation for the secondary network from that of the primary network. In \cite{Chandramouli} and \cite{Le} the authors consider the power allocation for secondary users under maximum tolerable interference constraint on the primary users, besides minimum QoS constraint on each secondary user. They address two scenarios, the first is the maximization of the secondary network throughput in case that all of the secondary users can be supported under the given constraints. The second is the maximization of the number of admitted secondary users under their minimum QoS constraints when not all of them can be supported. The second scenario is also addressed in \cite{Globecom} where a lower complexity and more efficient algorithm is proposed. In \cite{Joint} the authors combine the power allocation and admission control for secondary users in one problem aiming at maximizing the overall secondary throughput.

In this work we consider an underutilized primary network and investigate the overall network throughput in the presence of cognitive radios. Our objective is to jointly allocate power to both primary and secondary users so that the overall network throughput is maximized, while protecting the QoS of the primary users by keeping their signal-to-interference-and-noise-ratio (SINR) above a minimum level. Since power allocation for maximum sum throughput in multiuser wireless networks is a non-convex problem, we introduce an iterative algorithm that is proved to converge to an optimal value (not necessarily the global) of the total sum throughput. We investigate the efficiency of the solutions to which the proposed algorithm converges via numerical simulations, and show that the network can achieve better utilization by incorporating secondary users operation. Also, we show the throughput and power prices paid by primary users in order to sustain the secondary operation and achieve higher utilization. Our results show that although the primary users may lose throughput due to the secondary operation, their total transmit power is significantly reduced.

The rest of this paper is organized as follows. In Section \ref{sec:SysMdl} we describe the underlay network model. In Section \ref{sec:ProbForm} we state the problem formulation. We introduce the proposed algorithm and establish its convergence in Section \ref{sec:IterAlgo}. Simulation results are presented in Section \ref{sec:Results} and the work is concluded in Section \ref{sec:Conc}.

\section{System Model}
\label{sec:SysMdl}

We consider a spectrum sharing model where secondary users share the same frequency band with primary users. In this model we focus on the uplink transmission from primary users toward their base station (BS), whereas the secondary users are communicating with each other in an ad-hoc fashion as shown in Fig. \ref{fig:SysMdl}.

\begin{figure}[htbp]
	\centering
		\includegraphics[width=.40\textwidth]{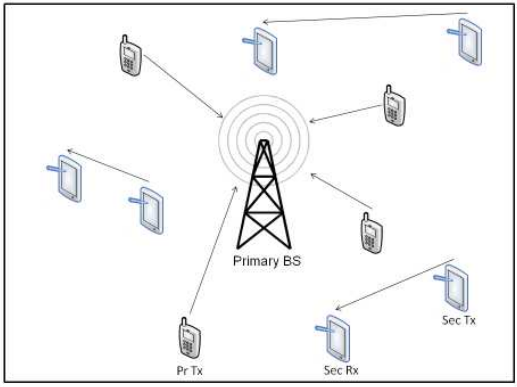}
	\caption{System Model}
	\label{fig:SysMdl}
\end{figure}

The primary network is assumed to have a single BS with $K$ primary users. We denote the transmit power of the $j^{th}$ primary user by $p_j^{p}$, $j=1,2,\cdots,K$, the channel gain from user $j$ to the base station by $g_{0,j}^{pp}$ and the SINR of primary user $j$ by $x_j^{p}$. Thus,

\begin{equation}
\label{eq:SINRp}
x_j^{p}=\frac{p_j^{p}g_{0,j}^{pp}}{\sum_{l=0,l\neq{j}}^{K}{p_l^{p}g_{0,l}^{pp}}+\sum_{i=0}^{N}{p_i^{s}g_{0,i}^{ps}}+N_0}, \ \ j=1,2,\cdots,K
\end{equation}
where there exist $N$ secondary links within the coverage area of the primary BS. Each secondary link $i$ comprises one transmitter and one corresponding receiver. The parameter $p_i^{s}$ is the transmit power of the transmitter of secondary link $i$, $g_{0,i}^{ps}$ is the channel gain between the $i^{th}$ secondary transmitter and the primary BS, and $N_0$ is the background noise power in the frequency band of operation. The channel gain between a secondary transmitter $i$ and a secondary receiver $j$ is $g_{j,i}^{ss}$, whereas the channel gain between a primary transmitter $i$ and a secondary receiver $j$ is $g_{j,i}^{sp}$. Hence, we can write the SINR of the secondary link $i$, $x_i^{s}$, as
\begin{equation}
\label{eq:SINRs}
x_i^{s}=\frac{p_i^{s}g_{i,i}^{ss}}{\sum_{l=0,l\neq{i}}^{N}{p_l^{s}g_{i,l}^{ss}}+\sum_{j=0}^{K}{p_j^{p}g_{i,j}^{sp}}+N_0}, \ \ i=1,2,\cdots,N
\end{equation}

A central controller is assumed to be responsible for the resource allocation for the network nodes. This controller is assumed to have knowledge of all of the channels in the system.

\section{Problem Formulation}
\label{sec:ProbForm}
For the given system, we consider the sum throughput. The network users are allocated powers so that the total throughput, primary and secondary, is maximized. However, in order to guarantee minimum QoS for the primary users, the power control scheme should guarantee that the SINR of each primary user does not fall below a minimum level. The QoS constraint on the primary users are
\begin{equation}
\label{eq:QoS_cons}
x^p_j\geq{\gamma^p_j}, \ \ \ \ j=1,2,\cdots,K
\end{equation}
This is equivalent to guaranteeing a minimum rate for each primary user.
Moreover, each transmitter, secondary or primary, cannot increase its transmit power beyond a certain maximum level. The power constraints on secondary and primary users are
\begin{equation}
\label{eq:Ps_cons}
p_i^s\leq{p_i^{smax}}, \ \ \ \ i=1,2,\cdots,N
\end{equation}
\begin{equation}
\label{eq:Pp_cons}
p_j^p\leq{p_j^{pmax}}, \ \ \ \ j=1,2,\cdots,K
\end{equation}

Using the definition of the throughput of user $i$ in bps/Hz as $\log_2{(1+SINR_i)}$, the optimization problem is written as

\begin{center}
\begin{equation}
\label{eq:Opt1}
\begin{tabular}{ r l }

\(\mbox{maximize}\) & \(\alpha_s\sum_{i=1}^{N}{\log_2{(1+x_i^s)}}+\alpha_p\sum_{j=1}^{K}{\log_2{(1+x_j^p)}}\) \\
\\
 \(\mbox{subject to}\) & \(x_j^p\geq{\gamma_j^p}, \	\ \ \ \ \ \ \ \ \ \ j=1,2,\cdots,K\) \\
\\
& \(0 \leq{p_i^s} \leq{p_i^{smax}}, \    \ i=1,2,\cdots,N\) \\
\\
& \(0 \leq{p_j^p} \leq{p_j^{pmax}}, \    \ j=1,2,\cdots,K\)
\end{tabular}
\end{equation}
\end{center}
where the optimization variables are $p_j^p$ and $p_i^s$ for $j=1,\cdots,K$ and $i=1,\cdots,N$. The SINRs: $x_j^p$ and $x_i^s$ are defined in \eqref{eq:SINRp} and \eqref{eq:SINRs}. Parameters $\alpha_s$ and $\alpha_p$ are control binary variables to switch between three different scenarios: 1) \emph{$\alpha_s=0$, $\alpha_p=1$} in which secondary users are not allowed to operate, 2) \emph{$\alpha_s=1$, $\alpha_p=1$} in which secondary users can share the spectrum with the primary users, and 3) \emph{$\alpha_s=1$, $\alpha_p=0$} in which the secondary sum throughput is maximized under the primary QoS constraints. We consider the second scenario in the next section. The optimization problem in \eqref{eq:Opt1} is a non-convex one and has no tractable solution \cite{Luo}. In \cite{Chandramouli}, a similar optimization problem is reduced to a convex problem, by omitting the unity term added to the SINR in each throughput term. That is, the system is assumed to operate in high SINR regime where $\log_2{(1+SINR_i)}\approx{\log_2{(SINR_i)}}$ and the problem is then converted to a geometric program and solved for a global optimum. But the difference in our formulation is that there are no constraints on the SINR of the secondary network, i.e., the system is opportunistic for the secondary users, consequently high SINR is not guaranteed for every user.

Also, the feasibility of the problem in (\ref{eq:Opt1}) necessitates that the primary network alone is feasible. That is, there exists a power allocation for every primary transmitter so that all the primary users can be supported exactly at their minimum QoS requirements. We assume that the proposed algorithm takes place when the primary network is feasible.

\section {Iterative Geometric Programming}
\label{sec:IterAlgo}

We convert the non-convex problem in \eqref{eq:Opt1} to a convex one by approximating the objective function by a monomial. We formulate the resulting problem as a geometric program that can be solved efficiently with polynomial computational complexity. Then, we iteratively update the approximate function and solve the geometric program again till convergence to at least a local maximum solution of the original problem.

\subsection{Geometric Programming} 
Before defining the geometric program we first introduce the following two terms:
\begin{itemize}
  \item A \emph{monomial} $m(\textbf{x})=Cx_1^{a_1}x_2^{a_2} \cdots x_n^{a_n}$ where $C$ is non-negative real constant, and $a_1, a_2, \cdots, a_n$ are real constants.
  \item A \emph{posynomial} $s(\textbf{x})=\sum_{k=1}^{K}{C_kx_1^{a_{1k}}x_2^{a_{2k}} \cdots x_n^{a_{nk}}}$ is a sum of monomials \cite{Boyd}.
\end{itemize}

A geometric program is defined as follows:

\begin{center}
\begin{equation}
\label{eq:OptGP}
\begin{tabular}{ r l }

\(\mbox{minimize}\) & \(f_0(\textbf{x})\) \\
 \(\mbox{subject to}\) & \(f_i(\textbf{x})\leq1, \    \ i=1,2,\cdots,l\) \\
& \(h_j(\textbf{x})=1, \    \ j=1,2,\cdots,m\)

\end{tabular}
\end{equation}
\end{center}
where, $\textbf{x}\in\Re_{++}^n$, $f_i(\textbf{x})$, $i=0,1,\cdots,l$ are posynomials and $h_j(\textbf{x})$, $j=1,2,\cdots,m$ are monomials.

Geometric programs can be simply converted into convex problems \cite{Boyd}. Since $\sum_{i=1}^{n}{\log_2{f(x_i)}}=\log_2{\prod_{i=1}^{n}{f(x_i)}}$, the optimization problem (\ref{eq:Opt1}) can be written as: 
\begin{center}
\begin{equation}
\label{eq:OptGP2}
\begin{tabular}{  l  }
\( \mbox{maximize}_{\textbf{p}^s,\textbf{p}^p,\textbf{x}^s,\textbf{x}^p} \ \ \prod_i^{N}{(1+x_{i=1}^s)\prod_{j=1}^{K}{(1+x_j^p)}} \)
\\
\\
 \(\mbox{subject to}\) \\
 \\
  \(\gamma_i^p{x_i^p}^{-1}\leq{1}, \    \ i=1,2,\cdots,K\) \\ 
\\  \(\frac{x_j^p(\sum_{k=1,k\neq{j}}^{K}{p_k^pg_{0,k}^{pp}}+\sum_{l=1}^{N}{p_l^sg_{0,l}^{ps}+N_0})}{(p_j^pg_{0,j}^{pp})}\leq{1}, \    \ j=1,2,\cdots,K\) \\
\\ \(\frac{x_i^s(\sum_{j=1,j\neq{i}}^{N}{p_j^sg_{i,j}^{ss}}+\sum_{k=1}^{K}{p_k^pg_{i,k}^{sp}+N_0})}{(p_i^sg_{i,i}^{ss})}\leq{1}, \    \ i=1,2,\cdots,N\) \\
\\
 \( p_j^p({p_j^{pmax}})^{-1} \leq{1}, j=1,2,\cdots,K\)\\
\\
 \(p_i^s({p_i^{smax}})^{-1}\leq{1}, i=1,2,\cdots,N \)\\
\end{tabular}
\end{equation}
\end{center}
where $\textbf{p}^s$ and $\textbf{p}^p$ are the vectors of secondary and primary power allocations, respectively, and $\textbf{x}^s$ and $\textbf{x}^p$ are vectors of SINR allocations of secondary and primary users. We note that the constraints of the optimization problem (\ref{eq:OptGP2}) satisfy the conditions of posynomials. However, the conversion of the objective from maximization to minimization while maintaining the posynomial form is not achievable. For this reason we approximate the objective function, $f_0(\textbf{x}^s,\textbf{x}^p)$, to a monomial $\tilde{f}_0(\textbf{x}^s,\textbf{x}^p)$ where
\begin{equation}
\label{eq:monomial_approx}
\tilde{f}_0(\textbf{x}^s,\textbf{x}^p)=c\prod_{i=1}^{N}{{(x_i^s)}^{\lambda_i^s}}\prod_{j=1}^{K}{(x_j^p)^{\lambda_j^p}}
\end{equation}
consequently, the approximate optimization problem will be 
\begin{center}
\begin{equation}
\label{eq:Opt_Approx}
\begin{tabular}{  l  }
\(\mbox{maximize}_{\textbf{p}^s,\textbf{p}^p,\textbf{x}^s,\textbf{x}^p} \ \ c\prod_{i=1}^{N}{(x_i^s)^{\lambda_i^s}}\prod_{j=1}^{K}{(x_j^p)^{\lambda_j^p}} \) 
\\
\\
 \(\mbox{subject to}\) \\
\\
\\
  \(\gamma_i^p{x_i^p}^{-1}\leq{1}, \    \ i=1,2,\cdots,K\) \\
\\ \(\frac{x_j^p(\sum_{k=1,k\neq{j}}^{K}{p_k^pg_{0,k}^{pp}}+\sum_{l=1}^{N}{p_l^sg_{0,l}^{ps}+N_0})}{(p_j^pg_{0,j}^{pp})}\leq{1}, \    \ j=1,2,\cdots,K\) \\
\\ \(\frac{x_i^s(\sum_{j=1,j\neq{i}}^{N}{p_j^sg_{i,j}^{ss}}+\sum_{k=1}^{K}{p_k^pg_{i,k}^{sp}+N_0})}{(p_i^sg_{i,i}^{ss})}\leq{1}, \    \ i=1,2,\cdots,N\) \\
\\
 \( p_j^p({p_j^{pmax}})^{-1} \leq{1}, j=1,2,\cdots,K\)\\
\\
 \(p_i^s({p_i^{smax}})^{-1}\leq{1}, i=1,2,\cdots,N \)\\
\end{tabular}
\end{equation}
\end{center}
where the parameters $c$, $\lambda_i^s$, and $\lambda_j^p$, $i=1,\cdots,N$ and $j=1,\cdots,K$ are determined iteratively so that, for any two successive iterations $k$ and $k-1$ we always have non-decreasing sum throughput, i.e.,
\begin{equation}
\label{eq:non_dec}
f_0(\textbf{x}_0^{s(k-1)},\textbf{x}_0^{p(k-1)}) \leq{f_0(\textbf{x}_0^{s(k)},\textbf{x}_0^{p(k)})}
\end{equation}
where $(\textbf{x}_0^{s(k)},\textbf{x}_0^{p(k)})$ is the solution to the approximate problem \eqref{eq:Opt_Approx} for the $k^{th}$ iteration.

\subsection{Determination of $\tilde{f}_0(\textbf{x}^s,\textbf{x}^p)$ for Non-decreasing Throughput}

In this subsection, we consider the complete determination of $\tilde{f}_0(\textbf{x}^s,\textbf{x}^p)$ in a certain iteration so that the inequality of \eqref{eq:non_dec} holds. Let $(\textbf{x}_0^{s},\textbf{x}_0^{p})$ be the solution to the approximate problem \eqref{eq:Opt_Approx} in the previous iteration. Then, to have a sum throughput in the following iteration that is at least no less than that corresponding to $(\textbf{x}_0^{s},\textbf{x}_0^{p})$, the approximate function  $\tilde{f}_0(\textbf{x}^s,\textbf{x}^p)$ should satisfy the following conditions:
\begin{enumerate}
\item $f_0(\textbf{x}^s,\textbf{x}^p)\geq{\tilde{f}_0(\textbf{x}^s,\textbf{x}^p)}$ for all $\textbf{x}^s,\textbf{x}^p$,
\item $f_0(\textbf{x}_0^s,\textbf{x}_0^p)=\tilde{f}_0(\textbf{x}_0^s,\textbf{x}_0^p)$,
\item $\nabla f_0(\textbf{x}_0^s,\textbf{x}_0^p)=\nabla \tilde{f}_0(\textbf{x}_0^s,\textbf{x}_0^p)$.
\end{enumerate}    

The above three conditions are sufficient to guarantee that the solution of each approximate problem increases the objective function, and after convergence of the series of the approximate problems, the Karush-Kuhn-Tucker conditions of the original problem in \eqref{eq:OptGP2} will be satisfied \cite{Chiang,Old}.
Let $\textbf{x}=[x_1, x_2 \cdots, x_{N+K}]$ where $x_j=x_j^p$, $j=1,2,\cdots,K$, $x_j=x_i^s$, $j=K+i$, $i=1,2,\cdots,N$, $\Lambda=[\lambda_1, \lambda_2, \cdots, \lambda_{N+K}]$ where $\lambda_j=\lambda_j^p$, $j=1,2,\cdots,K$, $\lambda_j=\lambda_i^s$, $j=K+i$ and $i=1,2,\cdots,N$.

From the third condition we have:
\begin{equation*}
\nabla f_0(\textbf{x}_0)=\nabla \tilde{f}_0(\textbf{x}_0)
\end{equation*}
Therefore,
\begin{equation}
\label{eq:lambda}
\lambda_i=\frac{x_{0i}}{1+x_{0i}}, \ \ \ \ i=1,2,\cdots,N+K 
\end{equation}

From the second condition we have:
\begin{equation*}
f_0(\textbf{x}_0)=\tilde{f}_0(\textbf{x}_0) \\ 
\end{equation*}
\begin{equation*}
c\prod_{i=1}^{N+K}{{x_{0i}}^{\lambda_i}}=\prod_{i=1}^{N+K}{(1+x_{0i})}
\end{equation*}
\begin{equation}
\label{eq:c}
c=\frac{\prod_{i=1}^{N+K}{(1+x_{0i})}}{\prod_{i=1}^{N+K}{(x_{0i})}^{\lambda_i}}
\end{equation}

By obtaining $c$, $\lambda_i^s$ and $\lambda_j^p$ for $i=1,\cdots,N$ and $j=1,\cdots,K$, the approximate objective function $\tilde{f}(\textbf{x}_0^s,\textbf{x}_0^p)$ is determined for the next iteration. What remains is to show how would the new approximate function satisfy the first condition. It is required to prove that
\begin{equation*}
\prod_{i=1}^{N+K}{(1+x_i)}\geq{c\prod_{i=1}^{N+K}{x_i^{\lambda_i}}}
\end{equation*}
Equivalently by substituting for $c$ and $\lambda_i$ from \eqref{eq:c} and \eqref{eq:lambda},
\begin{equation}
\label{eq:Ineq}
\prod_{i=1}^{N+K}{(1+x_i)}\geq{\prod_{i=1}^{K+N}{\frac{(1+x_{0i})}{x_{0i}^{(\frac{x_{0i}}{1+x_{0i}})}}.x_i^{(\frac{x_{0i}}{1+x_{0i}})}}}
\end{equation}

\begin{theorem}
Fix $i$, for all $x_i\geq{0}$ the following inequality is satisfied:
\begin{equation}
\label{eq:Fix_i}
1+x_i\geq{\frac{(1+x_{0i})}{x_{0i}^{(\frac{x_{0i}}{1+x_{0i}})}}.x_i^{(\frac{x_{0i}}{1+x_{0i}})}}, \ \ \ x_i\geq{0}
\end{equation}
\end{theorem}
\begin{proof}
By combining both sides of \eqref{eq:Fix_i}, and taking the $\log$ we define:
\begin{equation}
\label{eq:log}
g(x_i)=\log{\frac{1+x_{0i}}{1+x_i}}+\frac{x_{0i}}{1+x_{0i}}\log{\frac{x_i}{x_{0i}}}
\end{equation}
Now we prove that $g(x_i)\leq{0}$ by showing that the maximum of $g(x_i)$ is zero. First, the stationary point is obtained by:
\begin{equation*}
\frac{\partial{g(x_i)}}{\partial{x_i}}|_{x_i=x^*_i}=0
\end{equation*}
\begin{equation*}
\frac{-1}{1+x^*_i}+\frac{x_{0i}}{x_i^*(1+x_{0i})}=0 \ \ \Rightarrow{\ \ x^*_i=x_{0i}}
\end{equation*}
By substituting with $x_i^*$ in \eqref{eq:log}:
\begin{equation*}
g(x_i^*=x_{0i})=0
\end{equation*}
The second derivative of $g(x_i)$ is:
\begin{equation*}
\frac{\partial^2{g(x_i^*)}}{\partial{x_i^{*^2}}}=\frac{-1/x_{0i}}{(1+x_{0i})^2}
\end{equation*}
which is negative for $x_{0i}\geq{0}$.
Therefore, $g(x_i)\leq{0}$ for $x_i\geq{0}$, and consequently the inequality of \eqref{eq:Fix_i} is satisfied.

Since the inequality of \eqref{eq:Fix_i} is satisfied for all $i=1,2,\cdots,N+K$, inequality \eqref{eq:Ineq} is also satisfied, and the first condition on the approximate function is maintained.
\end{proof}


Our proposed iterative algorithm is implemented as follows:
\begin{enumerate}
\item Choose positive initial values for c,$\lambda_i$, $i=1,\cdots,N+K$.
\item Solve the approximate problem in \eqref{eq:Opt_Approx}.
\item If a target accuracy is reached, stop the algorithm. Otherwise go to step 4).
\item Calculate new values of $\lambda_i$ and $c$ from \eqref{eq:lambda}, \eqref{eq:c}, then go to step 2).
\end{enumerate}

Note that, the proposed technique can be applied to any sum throughput maximization problem where the constraints are convex.
In \cite{Sequential}, an approach to solve the sum throughput maximization in interference channels is proposed based on single and double condensation methods \cite{Chiang}. The single condensation method is shown to be outperformed by the double condensation method. However the approximate objective function of the double condensation method was not proved to satisfy the first condition in Section. \ref{sec:IterAlgo}.

\section{Numerical Results}
\label{sec:Results}
The simulation setup we use is close to that used in \cite{Le}. We assume that the secondary transmitters and the primary users are located in an area of size $1000m\times{1000}m$ with BS of the primary network located at the center. The receiving node of each secondary link is placed randomly in a $500m\times{500}m$ square with its transmitting node at the center. The channel gain between any transmitter $j$ and receiver $i$ is modeled as $g_{i,j}=K_0\times 10^{\beta_{i,j}/10}\times d_{i,j}^{-4}$, where $d_{i,j}$ is the corresponding distance, $\beta_{i,j}$ is a random Gaussian variable with zero mean and a standard deviation of 6 dB to account for shadowing effects, and $K_0=10^3$ is a factor that includes some system parameters such as antenna gain and carrier frequency. The background noise power is $N_0=10^{-12}$ watt. The maximum transmit power on each secondary transmitter is $p^{smax}=0.1$ watt, and on each primary transmitter $p^{pmax}=1$ watt. For simulations we use the \texttt{CVX}, a package for specifying and solving convex programs \cite{CVX}.
\subsection{Convergence of the Iterative Algorithm}
We compare the value of the sum throughput to which the proposed algorithm converges with the optimum solution obtained by exhaustive search, for this simulation we consider $N=2$ and $K=2$, the minimum QoS for the primary transmitters is chosen to be $\gamma^p=6$ dB. For the feasibility of the primary network we assume primary and secondary users employ CDMA with primary processing gain $PG^p=20$, and secondary processing gain $PG^s=80$. We run five different simulations where in each one the locations of all users are generated randomly. The exhaustive search result is obtained by discretizing the transmit power range of each user into $100$ samples and trying all possible combinations of these samples, then picking the one that achieves the maximum sum throughput under the QoS constraints. We start the iterative algorithm with all exponents in $\tilde{f}$ equal to one, i.e., $\Lambda=[1, 1, 1, 1]$.
Fig. \ref{fig:ExComp} shows how the iterative algorithm almost converges to the optimum solution for small number of users. The proposed algorithm, however, is not guaranteed to converge to the global solution in general situations. 
\begin{figure}
	\centering
		\includegraphics[width=0.45\textwidth]{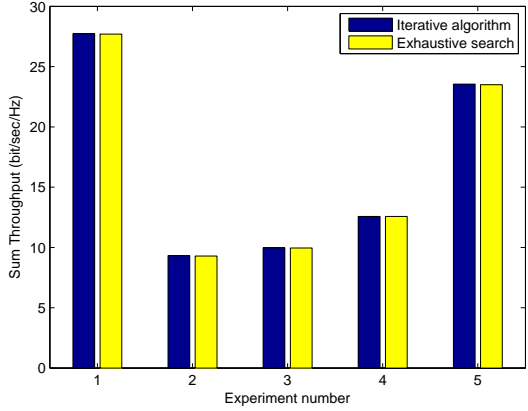}
	\caption{Total network throughput in different simulations for both the iterative and exhaustive search schemes $N=2$ and $K=2$}
	\label{fig:ExComp}
\end{figure}
The convergence in a low SINR regime is shown in Fig. \ref{fig:Convergence}, where a simulation is run under unity processing gain for both primary and secondary users. And the minimum QoS for the primary users is $\gamma^p=-3$ dB, while the initial condition is not changed. Here, we note that, the initial condition of equal values of $\lambda$ means that the system is at high SINR which is not true in the current scenario, and that is obvious from the solution of the first iteration, however the proposed algorithm then iteratively increases the throughput till it reaches a local maximum that is the exhaustive search solution in this setup. 
\begin{figure}
	\centering
		\includegraphics[width=0.45\textwidth]{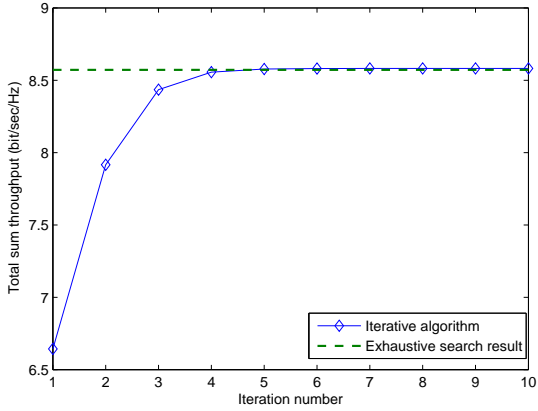}
	\caption{Convergence of the proposed algorithm in  low SINR regime for $N=2$ and $K=2$}
	\label{fig:Convergence}
\end{figure}

\subsection{Network Utilization and Primary Cost}
In this subsection we compare between the total throughput of the primary network alone, and the total throughput of the overall underlay network. In this simulation we assume $N=3$ users, and $K=5$. Both primary and secondary users use processing gain of $80$. The minimum QoS for each primary user is $\gamma^p=8$ dB. We use the proposed iterative algorithm in the three scenarios discussed in Section \ref{sec:ProbForm}. The initial components of the vector $\Lambda$ are equal. We run five experiments, where in each one, the locations of all users are generated randomly and the three optimization scenarios take place. For each optimization scenario we calculate the total sum throughput, the primary sum throughput and the total primary transmit power. The results showing those three quantities are depicted in Fig. \ref{fig:Th_tot}, Fig. \ref{fig:Th_pri} and Fig. \ref{fig:P_pri}, respectively.
\begin{figure}
	\centering
		\includegraphics[width=0.45\textwidth]{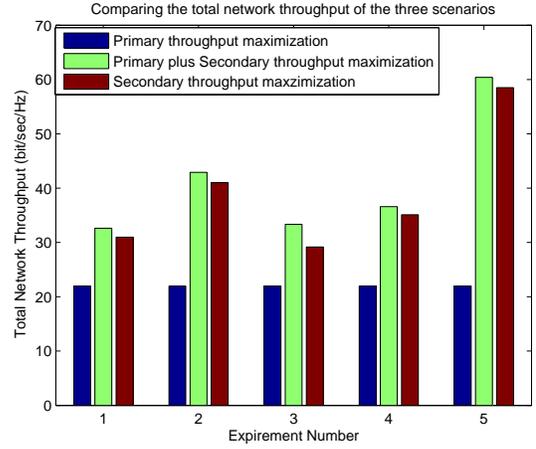}
	\caption{Total underlay network throughput for $N=3$ and $K=5$}
	\label{fig:Th_tot}
\end{figure}
\begin{figure}
	\centering
		\includegraphics[width=0.45\textwidth]{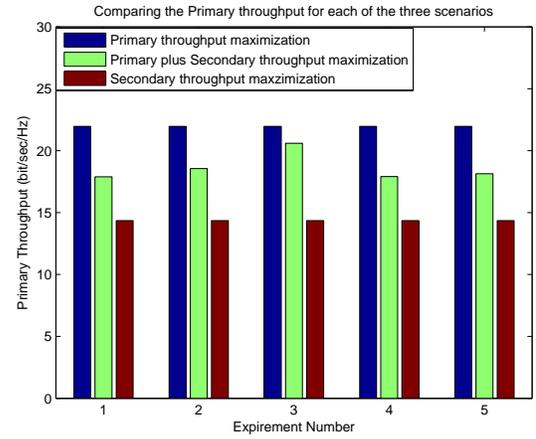}
	\caption{Primary network sum throughput for $N=3$ and $K=5$}
	\label{fig:Th_pri}
\end{figure}
\begin{figure}
	\centering
		\includegraphics[width=0.45\textwidth]{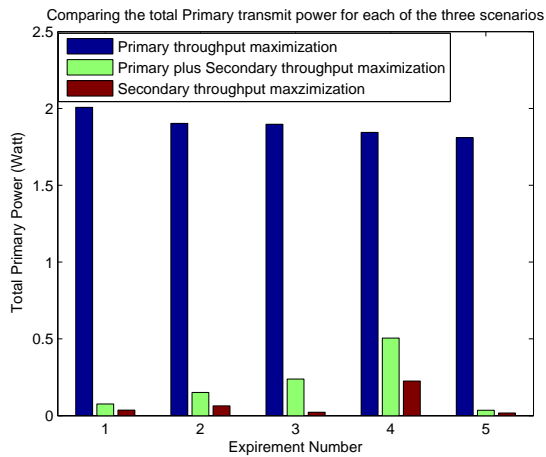}
	\caption{Total primary transmit power for $N=3$ and $K=5$}
	\label{fig:P_pri}
\end{figure}
 From Fig. \ref{fig:Th_tot} it is clear that the second scenario in which the primary-plus-secondary sum throughput is maximized is the one that leads to higher utilization in terms of total underlay network throughput. In Fig. \ref{fig:Th_pri} the primary throughput is shown, and it is clear that the concurrent operation of the secondary and primary users leads to degradation of the throughput that the primary network can achieve with no secondary operation. However, the throughput loss of the primary users in the second scenario is not severe and a cooperative primary network may tolerate this throughput price for the sake of more utilization. Fig. \ref{fig:P_pri} shows the total transmit power of the primary network in each of the three scenarios. Although the primary users lose some throughput due to the secondary operation, their total transmit power is reduced significantly as a consequence of the throughput loss. Since the third scenario yields minimum transmit primary power under maximization of the secondary network throughput, the third scenario can be used in an optimization problem that aims at maximizing the secondary sum throughput while the primary users are kept at their minimum possible transmit power with QoS constraints.

\section{Conclusion}
\label{sec:Conc}
We have considered the resource allocation for spectrum underlay cognitive radio networks. Our objective is to allocate power jointly to primary and secondary users so that the total underlay network throughput is maximized. In the optimization problem the primary users are protected against the opportunistic behavior by setting a minimum QoS constraint for each of them. Since the objective function of the optimization problem is non-convex, we introduce an iterative algorithm based on successive solutions to approximate geometric programs. We provide a mathematical proof of the convergence of the proposed algorithm and show that the algorithm converges always to one of the local optima of the original problem. Then we show via simulations the efficient results to which the algorithm converges compared to the exhaustive search method. Finally, we show how the algorithm is used to increase the overall network throughput (utilization) with relatively small loss in primary throughput and with considerable decrease in total primary transmit power.

\end{document}